\documentclass[letterpaper,12pt]{amsart}
\usepackage[section]{algorithm}
\usepackage{latexsym,array,delarray,amsthm,amssymb,epsfig}

\theoremstyle{plain}
\newtheorem{thm}{Theorem}[section]
\newtheorem{lemma}[thm]{Lemma}
\newtheorem{prop}[thm]{Proposition}
\newtheorem{cor}[thm]{Corollary}

\newtheorem*{thm*}{Theorem}
\newtheorem*{lemma*}{Lemma}
\newtheorem*{prop*}{Proposition}
\newtheorem*{cor*}{Corollary}
\newtheorem*{conj*}{Conjecture}

\theoremstyle{definition}
\newtheorem{defn}[thm]{Definition}
\newtheorem{ex}[thm]{Example}
\newtheorem{pr}[thm]{Problem}

\theoremstyle{remark}

\newcommand{\rr}{\mathbb{R}}

\newcommand{\Span}{\mathrm{span}}
\newcommand{\ind}{\mbox{$\perp \kern-5.5pt \perp$}}

\makeatletter
\renewcommand*\env@matrix[1][*\c@MaxMatrixCols c]{%
  \hskip -\arraycolsep
  \let\@ifnextchar\new@ifnextchar
  \array{#1}}
\makeatother

\begin{document}

\title{Distance-based phylogenetic methods around a polytomy}
\author{Ruth Davidson and Seth Sullivant}
\maketitle

\begin{abstract}
Distance-based phylogenetic algorithms attempt to solve the NP-hard 
least squares phylogeny problem by mapping an arbitrary 
dissimilarity map representing biological data to a tree metric.  
The set of all dissimilarity maps is a Euclidean space properly 
containing the space of all tree metrics as a polyhedral fan.   
Outputs of distance-based tree reconstruction algorithms such as 
UPGMA and Neighbor-Joining are points in the maximal cones in the fan.   
Tree metrics with polytomies lie at the intersections of maximal cones. 

A phylogenetic algorithm divides the space of all dissimilarity maps
into regions based upon which combinatorial tree is 
reconstructed by the algorithm.  Comparison of phylogenetic methods
can be done by comparing the geometry of these regions.
We use polyhedral geometry to compare the local nature of the subdivisions
induced by least squares phylogeny, UPGMA, and Neighbor-Joining.
Our results suggest that in some circumstances, UPGMA and Neighbor-Joining
poorly match least squares phylogeny when the true tree has a polytomy.
\end{abstract}


\section{Introduction}
A function $\alpha : X \times X  \to \rr$ is called a 
\textit{dissimilarity map} if for all $x, y \in X$, $\alpha(x,x) = 0$ and 
$\alpha(x,y) = \alpha(y,x)$.  A dissimilarity map $\alpha$ is a \emph{tree
metric} if it arises as the set of pairwise distances between
the leaves in a tree with edge lengths.
A distance-based phylogenetic method is a procedure that
takes as input a dissimilarity map $\alpha$ and returns
a tree metric $\hat \alpha$.

Among the most intuitively appealing distance-based phylogenetic
methods is the  \emph{least-squares phylogeny} (LSP).  The least-squares phylogeny problem asks, for a given dissimilarity $\alpha$,
what is the tree metric $\hat\alpha$ that minimizes the Euclidean distance: 
$$
d_{2}(\alpha, \hat\alpha) :=  \sqrt{\sum_{x,y \in X} 
(\alpha(x,y) - \hat\alpha(x,y))^{2} }. 
$$
The least squares phylogeny problem is NP-hard \cite{Day1987},
and because of this many distance-based phylogenetic algorithms have
been developed which attempt to build up the tree piece by
piece while locally optimizing the Euclidean distance at
each step.  Two popular agglomerative distance based-methods designed according to
this philosophy are
UPGMA (Unweighted Pair-Group Method with Arithmetic Mean) and NJ (Neighbor-Joining), which both run in polynomial time.  Since LSP is NP-hard,
UPGMA and NJ cannot solve LSP exactly.  So it is natural to ask:
how well do these distance-based methods perform when attempting
to solve the LSP problem?  Under what circumstances do
distance-based heuristics return the same combinatorial tree as
the least squares phylogeny?

A well-known consistency result of Atteson \cite{Atteson1997}
says the following:  Let $\alpha$ be a tree metric, arising from a 
binary tree 
all of whose branch lengths are bounded away from zero, and
let $\alpha'$ be a dissimilarity map which is sufficiently close to some tree metric
$\alpha$.  Then NJ applied to 
$\alpha'$ returns a tree with the same topology as $\alpha$.
A similar statement also holds for the LSP problem,
since other tree metrics with a different topology are necessarily 
bounded away from a given tree metric with a fixed topology and
large edge lengths.  Hence, Neighbor-Joining gives a tree topology
consistent with the Least-Squares Phylogeny when all edge lengths
are bounded away from zero.  This leads us to the main question of 
study in the present paper:

\begin{pr}
How do distance based-heuristics (UPGMA, NJ) compare to the LSP when the true tree metric has a polytomy?
\end{pr}

A \emph{polytomy} is a vertex in a tree with more than
three neighbors.  In a rooted tree, this represents a speciation
event where many different species were produced.  Polytomies arise in tree construction from collections of species
for which there is not enough data to decide which sequence of binary events
is most relevant.

Our comparison of different phylogenetic reconstruction methods is 
based on methods from geometry.  In particular, any distance-based
phylogenetic reconstruction method partitions the set of all
dissimilarity maps into regions indexed by the possible combinatorial
types of tree reconstructed by the method.  We can then
compare these regions for different methods.  In the case
of the distance-based heuristics (UPGMA, NJ), the resulting regions
are polyhedral cones.  For the LSP, the regions
are potentially more complicated semialgebraic sets (solutions to polynomial
inequalities).  The idea of comparing the two
distanced-based methods using (polyhedral) geometry already appears
in \cite{NJSphere} and \cite{Haws2011}, comparing Neighbor-Joining
to Balanced Minimum Evolution (BME).  

In previous work, we characterized the polyhedral subdivision induced
by the UPGMA algorithm \cite{Davidson2013}.  
While we do not yet know a complete description
of the regions induced by LSP, a local analysis of the performance of LSP and distance-based heuristics near a polytomy
can be done using polyhedral geometry.  The
resulting analysis depends heavily on the geometry of
phylogenetic tree space near tree metrics that contain a polytomy.
It is this analysis which comprises the bulk of this paper. 

This paper is organized as follows:  in the next section, we
review basic properties of tree space, including a description
of the different cones in the standard decomposition.
We provide the description of both tree space and equidistant tree space.
Section \ref{sec:tritomygeometry} contains a detailed analysis
of the local geometry of tree space near tree metrics which
have a tritomy.  In particular, for both equidistant and ordinary
tree metrics, the local geometry depends only on the sizes
of the daughter clades around the tritomy, and not the particular
tree structure of those daughter clades.  In Section
\ref{sec:tritomymethods}, we apply the results from Section 
\ref{sec:tritomygeometry} to understand the local geometry of the
decompositions induced by LSP and UPGMA near tree metrics
that contain a tritomy.  We also explain why these results imply
that UPGMA poorly matches LSP in some circumstances, and we discuss
computational evidence towards the study of NJ from this perspective.
Section \ref{sec:conclusion} contains concluding remarks primarily
about the possibility of extending results for NJ.


\section{Tree space}

Our analysis of phylogenetic algorithms near a polytomy
depends heavily on the geometry of tree space.
Our goal in this section is to recall various notions about
tree space and some of the basic known properties about it
that we will use in subsequent section.  In particular,
by tree space, we will mean the set of all tree metrics
for a given fixed number of leaves $n$, and we will
distinguish between all tree metrics and equidistant tree
metrics.  We assume familiarity with combinatorial phylogenetics
\cite{Felsenstein, Semple2003} and polyhedral geometry \cite{Ziegler1995}.

Let $T$ be a tree with leaves labeled by a set $X$ with $n$ elements,
and let $w : E(T) \rightarrow \rr_{\geq 0}$ be
a function that assigns weights to the edges of $T$.
The \emph{tree metric} $d_{T,w}$ induced by $T$ and $w$ is the
dissimilarity map that assigns a distance $d_{T,w}(x,y)$ as the sum
of the weights along the unique path connecting $x$ and $y$ in $T$.
A tree metric is called an \emph{equidistant} tree metric if
there is a point on the tree, the root $\rho$, such that the
distance between $\rho$ and any leaf is the same.

The set of all dissimilarity maps is naturally identified
with $\rr^{n(n-1)/ 2}_{\geq 0}$, as ${n \choose 2} = n(n-1)/2$, and coordinates in this
space are indexed by unordered pairs of elements in $X$.
The set of tree metrics on $n$-leaf trees is a proper subset of $\rr^{n(n-1)/ 2}_{\geq 0}$, 
denoted $\mathcal{T}_{n}$ and called the \emph{space of trees}
or the space of tree metrics.  Similarly, the set of all
equidistant tree metrics is a subset of $\rr^{n(n-1)/ 2}_{\geq 0}$, is denoted $\mathcal{ET}_{n}$, and is called the
\emph{space of equidistant trees} or  the space of equidistant
tree metrics.  Note that these tree spaces differ from the
space studied in \cite{Billera2001}.

Both $\mathcal{T}_{n}$ and $\mathcal{ET}_{n}$ are
polyhedral fans.  That is, they are the unions of 
polyhedral cones, and when two cones meet, they meet
on common subfaces of both.  The space of trees $\mathcal{T}_{n}$
has one maximal cone for each unrooted trivalent tree.
The space of equidistant trees $\mathcal{ET}_{n}$
has one maximal cone for each rooted binary tree.
The extreme rays of these maximal cones are known in both cases.

\begin{defn}
For each $i \neq j \in X$, let $e_{ij}  \in \rr^{n(n-1)/ 2}$ be
the dissimilarity map such that $e_{ij}(i,j) = 1$ and $e_{ij}(x,y) = 0$
for all other pairs $x,y$.
Let $A_{1},A_{2}, \ldots, A_{k}$ be a collection of disjoint
subsets of $X$.  Define the dissimilarity map
$\delta_{A_{1}|A_{2}|\cdots | A_{k}}$
$$
\delta_{A_{1}|A_{2}|\cdots | A_{k}}  =  \sum_{ij}  e_{ij}
$$
where the sum ranges over all unordered pairs $(i,j)$ such that $i$ and $j$
belong to different blocks.
\end{defn}

In the special case where $A|B$ is a partition of $X$, $A|B$ is usually called a 
\emph{split}. The resulting
dissimilarity map $\delta_{A|B}$ is called a \emph{cut-semimetric}
or \emph{split-psuedometric}.  Each edge in a tree $T$ induces
a split of the leaves of $T$ obtained from the partition of the leaves
that arises from removing the indicated edge.  The set of all
splits implied by a tree $T$ is denoted $\Sigma(T)$.

\begin{prop}\label{prop:extreme}
Let $T$ be a phylogenetic $X$-tree.  The set of all tree metrics
compatible with $T$ is a simplicial cone, whose extreme rays are
the set of vectors $\{\delta_{A|B} : A|B \in \Sigma(T) \}.$
\end{prop}

This is a polyhedral geometry rewording of Theorem 7.1.8 of \cite{Semple2003}.
Note that the description from Proposition \ref{prop:extreme}
holds regardless of whether or not the tree $T$ is binary.
In particular, we see that the intersection of all cones associated
to a collection of trees corresponds to the cone associated to the
tree obtained from a common coarsening of all trees in the given collection.

The cones of the space of equidistant trees $\mathcal{ET}_{n}$
are not simplicial in general, but they can be subdivided
into cones based on ranked trees, which are simplicial.
We describe these cones now.  A \emph{ranked} tree is a rooted phylogenetic
$X$-tree with a rank ordering on the internal vertices.
When $X =[n]$, these are naturally in bijection with maximal chains in the lattice
of set partitions $\Pi_{n}$.

 \begin{figure}[h]
 \centering
 \includegraphics[width=10cm]{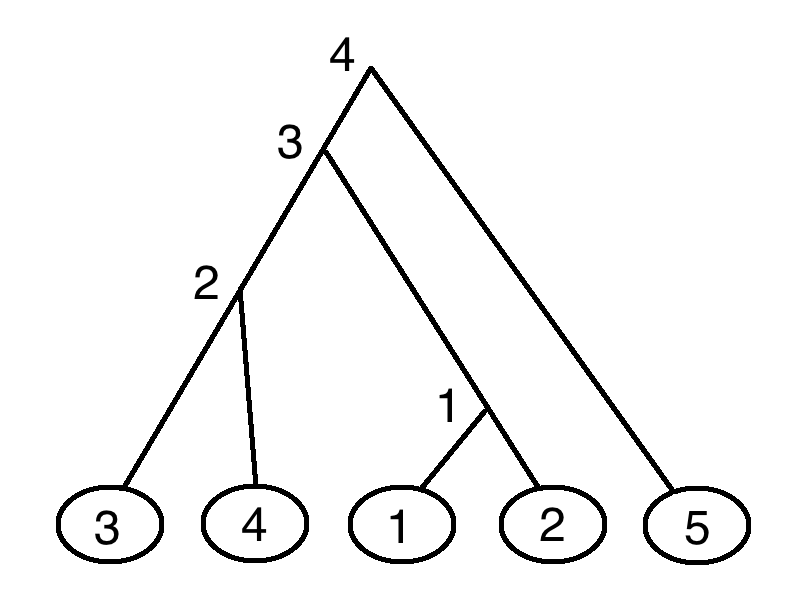}
 \caption{A ranked tree with leaf set $X = \{1,2,3,4, 5\}$.}
 \label{Rank_Example}
 \end{figure}
 
\begin{ex}
\label{RankEx}
The maximal chain
$$
\pi_{5} = 1|2|3|4|5 \lessdot 12|3|4|5 \lessdot 12|34|5 \lessdot 1234|5 \lessdot 12345 = \pi_{1}
$$
 in $\Pi_{5}$ corresponds to the ranked tree in Figure \ref{Rank_Example}.
\end{ex}

\begin{prop}\label{prop:equiextreme}
Let $$
C = \pi_{n} \lessdot \pi_{n-1} \lessdot \cdots \lessdot \pi_{1} 
$$ be a maximal chain in $\Pi_{n}$, corresponding
to a ranked phylogenetic tree.  The cone of equidistant tree metrics
compatible with $C$ is a simplicial cone whose extreme rays 
are the set of vectors
$\{\delta_{\pi_{i}} : i = 2, \ldots, n \}.$
\end{prop}

This is a polyhedral geometry rewording of Theorem 7.2.8 of \cite{Semple2003}.

\begin{ex}
If an equidistant tree metric $d \in \rr^{10}_{\geq 0}$  is compatible with the maximal chain in Example \ref{RankEx} then  $d$ satisfies
$$
d_{1,2} \leq d_{3,4} \leq d_{1,3} = d_{1,4} = d_{2,3} = d_{2,4} \leq d_{1,5} = d_{2,5} = d_{3,5} = d_{4,5}
$$
and is in the simplicial cone with extreme rays
$$
(1,1,1,1,1,1,1,1,1,1), (0,1,1,1,1,1,1,1,1,1),
$$
$$
(0,1,1,1,1,1,1,0,1,1), (0,0,0,1,0,0,1,0,1,1)
$$
where the coordinates of $\rr^{10}$ are labeled with the two-element sets $\{ i,j \} \in { [5] \choose 2 }$ in the lexicographic order.
\end{ex}

Note that Proposition \ref{prop:equiextreme} also holds true when
working with  chains that are not maximal, which correspond
to either trees with polytomies or situations where
there are ties in the rankings of the internal vertices.
These chains correspond to intersections of the maximal cones
associated to the maximal chains in the partition lattice.


\begin{figure}[h]
\centering
\includegraphics[width=12cm]{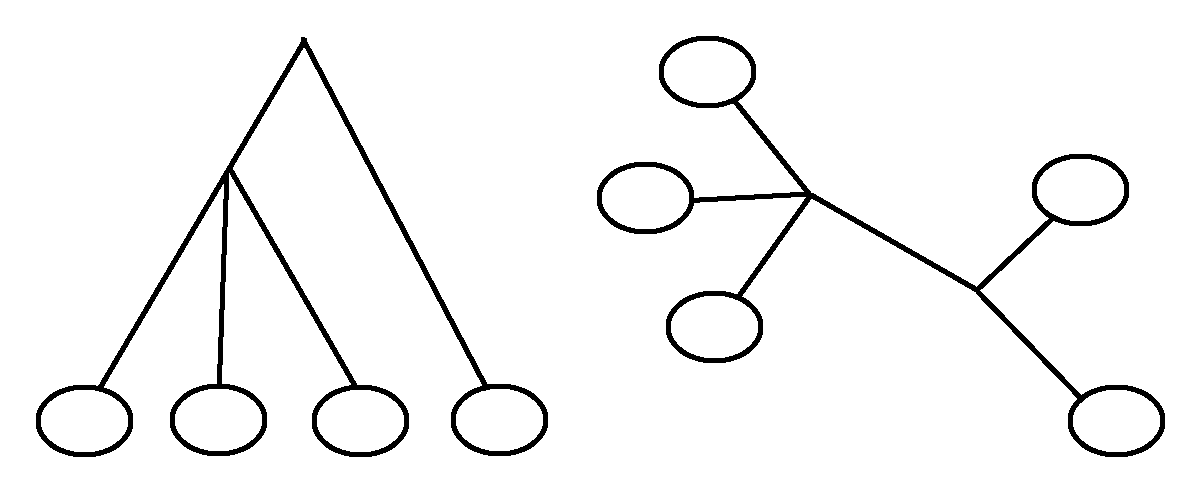}
\caption{Two tritomies, rooted and unrooted.}
\label{RootedUnRooted}
\end{figure}

\section{Geometry of Tree Space Near a Tritomy}\label{sec:tritomygeometry}

The goal of this section to describe the geometry of tree space near
a polytomy, in particular in the special case of tritomies.
For rooted trees, a \emph{tritomy} is an internal vertex that has
three direct descendants.  In an unrooted tree a \emph{tritomy}
is an internal vertex with four neighbors.
When we speak of the ``geometry of tree space near a tritomy'',
we mean to describe the geometry of tree space near a 
generic tree metric that is the tree metric of a tree with
a single tritomy and no other polytomies.  The set of all
such tritomy tree metrics, for a fixed topological
structure on the tree $T$, is a polyhedral cone of 
dimension one less than the dimension of tree space.
Let $C_{T}$ denote this polyhedral cone.
The tree $T$ with a single tritomy can be resolved
to three binary trees.  Denote them $T_{1}, T_{2}, T_{3}$.
The polyhedral cone of a tritomy $C_{T}$ is the intersection 
of the three \emph{resolution cones} $C_{T_{1}}$, $C_{T_{2}}$, and
$C_{T_{3}}$
associated to the
three different ways to resolve the tritomy tree into a
binary tree.

For both equidistant and ordinary tree space, the cones
$C_{T}$ and their resolution cones $C_{T_{i}}$, $i = 1,2,3$ 
satisfy  $\dim C_{T}  = \dim C_{T_{i}}  -1$.  
This is easily seen by the simplicial structure of the
cones $C_{T}$ for any tree $T$, according to Propositions
\ref{prop:extreme} and \ref{prop:equiextreme}.
Hence, locally near a generic point $x$ of $C_{T}$, tree space
looks like $\rr^{k}  \times K_{T}$ where $k = \dim C_{T}$
and $K_{T}$ is a one dimensional polyhedral fan that depends on
$T$ but does not depend on $x$.  
Furthermore, the fan $K_{T}$ can be chosen to live in
a space orthogonal to the span of $C_{T}$, and $\Span  \ K_{T}$ is two-dimensional. 
The goal of this
section is to describe the structure of that fan $K_{T}$.
The analysis depends on the particular structure of the
generators of the various cones involved, and the
cases of equidistant tree metrics and arbitrary tree metrics
must be handled separately.  We treat these cases in Sections
\ref{sec:equires} and \ref{sec:treeres}, respectively.

\begin{figure}
\centering
\includegraphics[width=7cm]{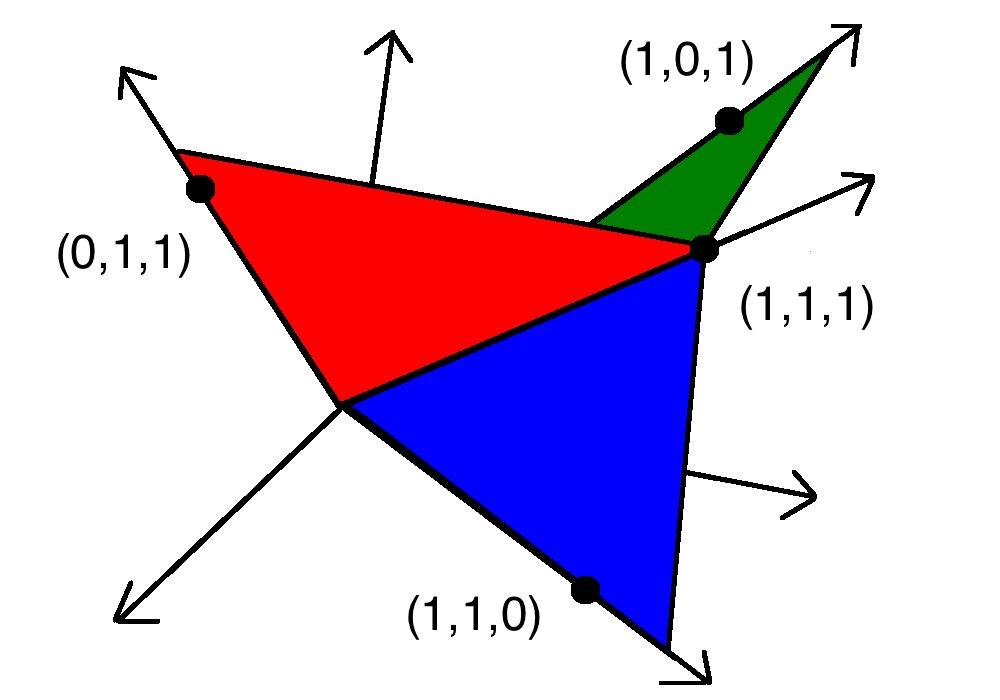}
\caption{The space $\mathcal{ET}_{3}$ with labeled extreme rays}
\label{open_book}
\end{figure}

\begin{figure}
\centering
\includegraphics[width=7cm]{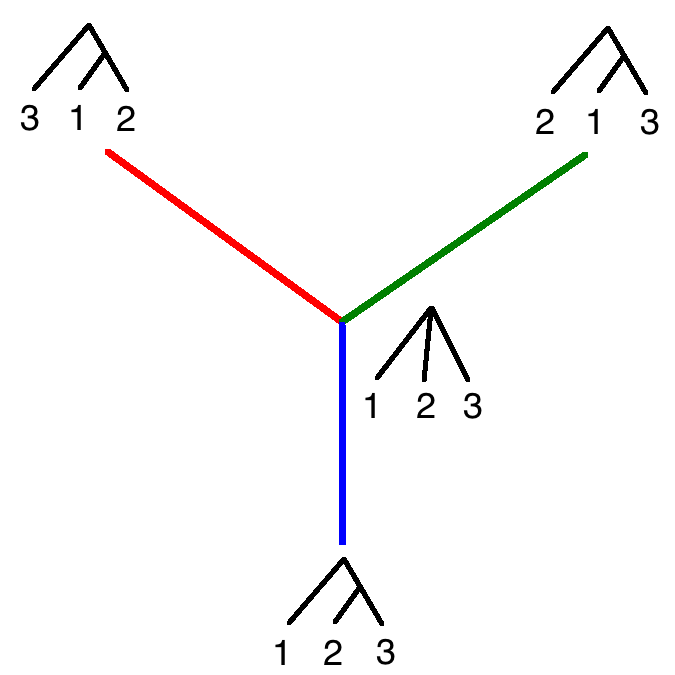}
\caption{The fan $K_{T}$ with labeled cones}
\label{KTFan}
\end{figure}



\subsection{Equidistant Tree Space}\label{sec:equires}

In this section we determine the geometry of the fan
$K_{T}$ for a tritomy tree $T$ in equidistant tree space.
This tritomy tree has a node with three children.  Denote
the daughter clades of these children (that is, the set
of leaves below each of children of the tritomy) by
$A$, $B$, and $C$.  Let $T_{AB}$, $T_{AC}$, and $T_{BC}$
denote the three resolution trees, where for example $T_{AB}$
is the binary resolution where $A \cup B$ forms a clade.
Note that since all the linear spaces that are involved
are the same, instead of working with a fixed tree we
can work with the corresponding rank function and chain in $\Pi_{n}$ to derive our results.  This is what we 
will do in this section.

Let 
$$
K = \pi_{n} \lessdot \pi_{n-1} \lessdot \cdots 
\lessdot \pi_{k+1} \lessdot \pi_{k-1} \lessdot  \cdots \lessdot \pi_{1}
$$
be the chain corresponding to the polytomy tree.  Note that
this is a chain in the partition lattice which leaves out an
element at the $k$-th level.
Here $\pi_{k+1}$ will contain among its blocks $A, B$ and $C$, and
$\pi_{k-1}$ will contain the block $A\cup B \cup C$.
  The resolution trees
$T_{AB}$, $T_{AC}$, and $T_{BC}$ correspond to the three ways
to add a $\pi_{k}$ to this sequence which refines $\pi_{k+1}$ and
is refined by $\pi_{k-1}$.

We are interested in the linear spaces $\Span \, C_{K}$. 

\begin{lemma}\label{span_basis}
For any  (not necessarily maximal) chain $K = \pi_{r} \lessdot \cdots \lessdot \pi_{1}$,
where $\pi_{1} = X$, 
the set of vectors
$$
\{\delta_{\pi_{i}} - \delta_{\pi_{i-1}} \}_{i = 2, \ldots, r }
$$
forms an orthogonal basis for $\Span \  C_{K}$.
\end{lemma}

\begin{proof}
Since $\delta_{\pi_{2}}, \ldots, \delta_{\pi_{r}}$ are the
extreme rays of the simplicial cone $C_{K}$, they are linearly
independent and hence span $\Span \ C_{K}$.  We can easily solve
for the vectors $\delta_{\pi_{2}}, \ldots, \delta_{\pi_{r}}$ given 
$\delta_{\pi_{i}} - \delta_{\pi_{i-1}}, i = 2, \ldots, r $
hence $\Span \  C_{K} = \Span \  \{\delta_{\pi_{i}} - \delta_{\pi_{i-1}} \}_{i = 2, \ldots, r }$. 
For all $i \in [r]$, the positions of the ones in $\delta_{\pi_{i-1}}$
are a subset of the positions of the ones in $\delta_{\pi_{i}}$.
This guarantees that $\delta_{\pi_{i}} - \delta_{\pi_{i-1}}$
and $\delta_{\pi_{j}} - \delta_{\pi_{j-1}}$ do not have
any nonzero entries in the same positions when $i \neq j$.  Hence these
vectors are orthogonal.  Note that $\delta_{\pi_{1}}$ is the
zero vector if we assume that $\pi_{1} = X$.
\end{proof}

The particular structure of the vectors $\delta_{\pi_{i+1}} - \delta_{\pi_{i}}$
will be useful in what follows.

\begin{lemma}\label{lem:delminusdel}
Let $\pi$ and $\tau$ be two set partitions such that $\pi$ is a refinement
of $\tau$.  Then
$$
\delta_{\pi} - \delta_{\tau} = \sum  e_{i,j}
$$
where $i,j$ are in different parts of $\pi$ and the same part of $\tau$.
\end{lemma}

\begin{proof}
Trivial from the definition of $\delta_{\pi}$.
\end{proof}

Now for each of the resolution cones, for example $C_{T_{AB}}$, 
there is a unique ray $p_{AB} \in \Span \,C_{T_{AB}}$ that is orthogonal
to $\Span \, C_{T}$.  We explain how to construct that ray now.


\begin{lemma}\label{outlier_projection}
Let $a = |A|$, $b = |B|$, and $c = |C|$.
The vector $p_{AB}$ is given by 
$$
p_{AB} =  
 -\frac{ac + bc }{ab + ac + bc} \delta_{A|B} +  \frac{ab}{ab+ ac + bc} (\delta_{A|C} + \delta_{B|C}).
 $$
\end{lemma}

\begin{proof}

It suffices to start with any vector $r_{AB} \in \Span \, C_{T_{AB}} 
\setminus \Span \, C_{T}$ and 
project it onto the orthogonal complement of $\Span \, C_{T}$.   
We assume the tree $T$ is represented by the chain
$$
K = \pi_{n} \lessdot \pi_{n-1} \lessdot \cdots \lessdot \pi_{k+1} \lessdot \pi_{k-1} \lessdot  \cdots \lessdot \pi_{1}
$$
and the tree $T_{AB}$ by the chain
$$
K_{AB} = \pi_{n} \lessdot \pi_{n-1} \lessdot \cdots \lessdot \pi_{k+1} \lessdot \pi_{k} \lessdot \pi_{k-1} \lessdot  \cdots \lessdot \pi_{1}.
$$
For our vector we choose $r_{AB} = \delta_{\pi_{k}} - \delta_{\pi_{k-1}}$.
This vector is clearly not in $\Span \, C_{K}$ since it involves $\delta_{\pi_{k}}$.
Furthermore, $r_{AB}$ is already orthogonal to all the vectors in 
the orthogonal basis for $\Span \ C_{T}$, except for the vector $\delta_{\pi_{k+1}} - 
\delta_{\pi_{k-1}}$.  Hence, we can project the vector $r_{AB}$ onto the complement
of the space spanned by $\delta_{\pi_{k+1}} - 
\delta_{\pi_{k-1}}$, this will be the same as projecting on the complement of 
$\Span \ C_{T}$.

Note that by Lemma \ref{lem:delminusdel} $r_{AB}  =  \delta_{A\cup B| C} = \delta_{A|C} + \delta_{B|C}$.  Similarly $\delta_{\pi_{k+1}} - 
\delta_{\pi_{k-1}}  =  \delta_{A|B|C} = \delta_{A|B} + \delta_{A|C} + \delta_{B|C}$.
So we want to project $\delta_{A|C} + \delta_{B|C}$ onto the orthogonal
complement of $\delta_{A|B} + \delta_{A|C} + \delta_{B|C}$.
To find $p_{AB}$ it is enough to compute the component of 
$\delta_{A|C} + \delta_{B|C}$ that is perpendicular to 
$\delta_{A|B} + \delta_{B|C} + \delta_{A|C}$, otherwise known as the 
vector rejection of $\delta_{B|C} + \delta_{A|C}$ from 
$\delta_{A|B} + \delta_{B|C} + \delta_{A|C}$.  
Due to the orthogonality of the vectors $\delta_{A|B}, \delta_{A|C}$, and $\delta_{B|C}$ and the fact that, for example, $(||\delta_{A|B}||_{2})^{2} =|A||B|$, we have 
$$
 \frac{(\delta_{A|C} + \delta_{B|C}) \cdot (\delta_{A|B} + \delta_{B|C} + \delta_{A|C})}{(\delta_{A|B} + \delta_{B|C} + \delta_{A|C}) \cdot (\delta_{A|B} + \delta_{B|C} + \delta_{A|C})} = \frac{(||\delta_{B|C}||_{2})^{2} + (||\delta_{A|C}||_{2})^{2}}{(||\delta_{A|B}||_{2})^{2} + (||\delta_{B|C}||_{2})^{2} + (||\delta_{A|C}||_{2})^{2}}
$$
  So the vector rejection of $\delta_{B|C} + \delta_{A|C}$ from 
$\delta_{A|B} + \delta_{B|C} + \delta_{A|C}$  becomes 
$$
\delta_{A|C} + \delta_{B|C} - \frac{ac + bc }{ab + bc + ac} (\delta_{A|B} + \delta_{B|C} + \delta_{A|C}) = 
$$
$$
 -\frac{ac + bc}{ab + ac + bc} \delta_{A|B} +  \frac{ab}{ab + ac + bc} (\delta_{A|C} + \delta_{B|C}).
 $$ 
\end{proof}

This explicit formula for $p_{AB}$ implies that $\Span \ K_{T}$ is 2-dimensional:

\begin{cor}\label{rooted:2dim}
The space $\Span \ K_{T}$ is 2-dimensional.  
\end{cor}

\begin{proof}
The formulae for $p_{AB}$, $p_{AC}$ and $p_{BC}$ given by Theorem \ref{outlier_projection} show that 
$p_{AB}$ and  $p_{AC}$ are not parallel, so that $\Span \ K_{T}$ is at least two-dimensional, but 
$$
p_{AB} + p_{AC} + p_{BC} = \mathbf{0},
$$
where $\mathbf{0}$ denotes the zero vector in $\rr^{n(n-1)/2}$.  So $\Span \ K_{T}$ is exactly two dimensional.   
\end{proof}

Corollary \ref{rooted:2dim} also follows from the fact that the set of all equidistant  tree metrics is a tropical variety in $\rr^{n(n-1)/2}$, as shown in \cite{TropGrass}.

\begin{thm}\label{equidistant:angle}   The angle between the cones $C_{T_{AC}}$ and $C_{T_{BC}}$ is
$$
 \arccos \left(   \frac{-c}{\sqrt{(a + c)(b + c)} }\right). 
$$
\end{thm}

\begin{proof}
We must calculate the angle between the vectors 
$p_{AC}$ and $p_{BC}$.  This is
$$
\arccos \left( \frac{p_{AC} \cdot p_{BC}}{\|p_{AC} \|_{2}  \| p_{BC} \|_{2}} \right).
$$
Now
$$p_{AC} =  
 -\frac{ab + bc }{ab + ac + bc} \delta_{A|C} +  \frac{ac}{ab+ ac + bc} (\delta_{A|B} + \delta_{B|C})$$
and
$$
p_{BC} =  
 -\frac{ab + ac }{ab + ac + bc} \delta_{B|C} +  \frac{bc}{ab+ ac + bc} (\delta_{A|B} + \delta_{A|C}).
$$
Thus, $p_{AC} \cdot p_{BC}$ is given by

\small
$$
-\frac{ab + bc }{ab + ac + bc} \cdot \frac{bc}{ab+ ac + bc}  \cdot \| \delta_{A|C}\|_{2}^{2}
+\frac{ac}{ab+ ac + bc} \cdot  \frac{bc}{ab+ ac + bc} \cdot  \| \delta_{A|B}\|_{2}^{2}
-\frac{ac}{ab+ ac + bc}  \cdot \frac{ab + ac }{ab + ac + bc} \cdot \| \delta_{B|C}\|_{2}^{2}
$$

$$
\normalsize =  \frac{-a^{2}b c}{ab + bc + ac}.
$$
Similar calculations show that 
$$
\| p_{AC} \|_{2}  = \sqrt{ \frac{ abc(a+c)}{ab + ac+bc}}  \mbox{ and  }  
\| p_{BC} \|_{2}  = \sqrt{ \frac{ abc(b+c)}{ab + ac+bc}}.
$$
Combining these pieces produces the formula in the Theorem.
\end{proof}


\subsection{Tree Space} \label{sec:treeres}

In this section we determine the geometry of the fan
$K_{T}$ for a tritomy tree $T$ in unrooted tree space $\mathcal{T}_{n}$. The approach is similar to the analysis for equidistant tree space, but the structure of tree space is more complicated. 
In particular, finding an orthogonal basis for the space spanned by the rays of the intersection cone for a tritomy is less straightforward.  

Recall that a tritomy $p$ in an unrooted tree is an internal vertex of degree four. The edges adjacent to $p$ induce a four-way set partition $A | B | C | D $ of $[n]$.
Let $T_{AB}$ denote the resolution tree in which there is an edge inducing the split $A \cup B | C \cup D$. Note that $T_{AB} = T_{CD}$.    So there are three resolutions $T_{AB}, T_{AC}$ and $T_{AD}$ of  $T$.  For the remainder of this section,  let $a = |A|$, $b = |B|$, $c = |C|$, and $d = |D|$.  Let $r_{AB} = \delta_{A \cup B | C \cup D}$.

\begin{lemma}\label{outlier_lemma_2}
Let $T$ be an unrooted tree with a tritomy and corresponding partition $A | B | C | D$ of $[n]$.   Then $r_{AB} \in \Span \ C_{T_{AB}} \setminus  \Span \ C_{T}$, and $\dim \Span \ C_{T} = 2n-4$.  
\end{lemma}

\begin{proof}
By Proposition \ref{prop:extreme} each extreme ray of $C_{T_{AB}}$ (equivalently, $T$) comes 
from a split induced by an edge  of $T_{AB}$
(equivalently, an edge of $T$).   
A binary unrooted tree on $n$ leaves has $2n-3$ edges.   So the cone
$C_{T_{AB}}$ has $2n-3$ rays, one for each internal edge of the 
tree $T_{AB}$.  By contracting the edge that induces the split $A \cup B | C \cup D$ 
for any pair we obtain $T$.   
Therefore $C_{T}$ has $2n-4$ extreme rays that correspond to the 
$2n-4$ internal edges of $T$.  Since $C_{T}$ is simplicial, $\dim \Span \ C_{T} = 2n-4$.  
\end{proof}

The projections $p_{AB}$, $p_{AC}$ and $p_{AD}$ of $r_{AB}$, $r_{AC}$ and $r_{AD}$ onto $(\Span \ C_{T})^{\perp}$ are the maximal cones in the fan $K_{T}$.     As in the previous section, we use an orthogonal basis of $\Span \ C_{T}$ to simplify the necessary calculations. 

The vectors in the set $\mathcal{U} = \{ \delta_{A | B \cup C \cup D}, \delta_{B | A \cup C \cup D}, \delta_{C | A \cup B \cup D}, \delta_{D | A \cup B \cup C} \}$ are extreme rays of $T$.  The elements of $\mathcal{U}$  correspond to the four edges in $T$ adjacent to the tritomy $p$.  We show in the next Lemma that to calculate $p_{AB}$ it is sufficient to calculate the projection of $r_{AB}$ onto  $(\Span \ \mathcal{U})^{\perp}$.  First we require some additional notation:  let $e = (u,v)$ be an edge of $T$ not adjacent to $p$ where $v$ is the internal vertex of $T$ on the path to $p$ from $e$.  Let $e' = (w,v)$ be the unique edge in $T$ satisfying the conditions $(i)$ $e \neq e'$ and $(ii)$ $w$ appears on the path from $v$ to $p$ in $T$ (note that it is possible that $w = p$).  Let $A_{e} | B_{e} $ be the split of $[n]$ induced by $e$ and let $A_{e'} | B_{e'}$ be the split of $[n]$ induced by $e'$.  Note $A_{e} \subsetneq A_{e'}$.   Let $a_{e} = |A_{e}|$ and $a_{e'} = |A_{e'}|$.  
Let
$$
\mathcal{V}  = \left\{\delta_{A_{e} | B_{e}} -\frac{a_{e}}{a_{e'} } \delta_{ A_{e'}  |  B_{e'}} \  : \  p \notin e = (u,v),   \ e' \  \text{satisfies} \ (i), (ii) \right\}. 
$$

\begin{lemma}\label{span_basis_2}
Every vector in $\mathcal{V}$ is orthogonal to $r_{AB}, r_{AC}$ and $r_{AD}$ and  $\mathcal{U} \cup \mathcal{V}$ is a basis for $\Span \ C_{T} $.  
\end{lemma}

\begin{proof}
Since $T$ has exactly one tritomy, $T$ has $2n-4$ edges.  When $n = 4$, $2n-4 = 4$.  In this case $|\mathcal{U} | =  \dim \Span \ C_{T}$, and $\mathcal{U}$ is a basis for $\Span \ C_{T}$.  So, assume $n > 4$, then $\mathcal{V}$ is not empty because we can find edges $e$ and $e'$ satisfying $p \notin (u,v) = e$ and $e'$ satisfying conditions $(i)$ and $(ii)$.  
We will first show that each element of $\mathcal{V}$ is orthogonal to $r_{AB}, r_{AC}$, and $r_{AD}$. 
Let $\nu \in \mathcal{V}$, then 
$$
\nu =   \delta_{A_{e} | B_{e}} -\frac{a_{e}}{a_{e'} } \delta_{ A_{e'}  |  B_{e'}}
$$
 Note that $A_{e'}$ is contained in one of $A, B, C$, and $D$.  Without loss of generality we may assume that $A_{e'} \subset A$.  
Then it follows directly from the structure of the summands in the vector $\nu$ that 
$$
r_{AB} \cdot \nu = ( a_{e'})(c + d) \left(-\frac{a_{e}}{a_{e'}} \right) + a_{e}(c + d)  = 0
$$
Similar calculations show that $r_{AC}$ and $r_{AD}$ are also orthogonal to $\nu$.  

We obtain $2n-8$ vectors in $\mathcal{V}$ because there are $(2n-4)-4$ edges in $T$ that do not induce vectors in $\mathcal{U}$. So $|\mathcal{U} \cup \mathcal{V}| = 2n-4$.  Since $\mathcal{U} \cup \mathcal{V}$ is comprised of vectors that are linear combinations of split-pseudometrics, $\Span \  \mathcal{U} \cup \mathcal{V} \subset \Span \ C_{T}$.    The set $\mathcal{U} \cup \mathcal{V}$ is also linearly independent since it
can be seen as an upper triangular transformation of the set of extreme
rays of $C_{T}$, which were independent. Thus $\mathcal{U} \cup \mathcal{V}$ is a basis for $\Span \ C_{T}$.
\end{proof}

Due to the structure of the vectors $r_{AB}$ and the elements of $\mathcal{U}$, $p_{AB}$ is constant on the coordinates for each $\delta_{U|V}$ and we can write  
\begin{equation}\label{proj:coeffs}
p_{AB} = \sum_{\{U,V \} \in { \{A, B, C, D \} \choose 2}} w(AB)_{U | V} \cdot  \delta_{U | V}
\end{equation}
the coefficients $w(AB)_{U|V}$ will facilitate computation of dot products and 2-norms.

\begin{thm}\label{unrooted_formulas}
The angle between the cones $C_{T_{AB}}$ and $C_{T_{AC}}$ is 
$$
 \arccos\left(- \frac{bc + ad}{\sqrt{(a + b) (a + c) (b + d) (c + d)}} \right)
$$

\end{thm}

\begin{proof} By Lemma \ref{span_basis_2}, to find $p_{AB}$, $p_{AC}$, and $p_{AD}$ it is sufficient to calculate the projection of $r_{AB}$, $r_{AC}$, and $r_{AD}$ onto $(\Span \ \mathcal{U})^{\perp}$. We will find $p_{AB}$ by calculating the coefficients $w(AB)_{U | V}$.  First, we construct a matrix $M_{AB}$ where $(M_{AB})_{i,j}$ is obtained (up to row operations) by taking the dot product of the vector indexing row $i$ and column $j$:

$$
M_{AB} = \begin{matrix}[c|ccccc] \quad & r_{AB} & \delta_{A | B \cup C \cup D} &  \delta_{B | A \cup C \cup D} & \delta_{C | A \cup B \cup D} & \delta_{D | A \cup B \cup C } \\ \hline \delta_{A | B \cup C \cup D} & c + d & b + c + d & b & c & d \\  \delta_{B | A \cup C \cup D} & c + d & a & a + c + d & c & d \\  \delta_{C | A \cup B \cup D} & a + b & a & b & a + b + d & d  \\  \delta_{D | A \cup B \cup C} & a + b & a & b & c & a + b + c \end{matrix}
$$
Next, let $K_{AB}$ be the matrix given below:
$$
K_{AB} = \begin{matrix}[c|ccccc] \quad & r_{AB} & \delta_{A | B \cup C \cup D} &  \delta_{B | A \cup C \cup D} & \delta_{C | A \cup B \cup D} & \delta_{D | A \cup B \cup C } \\ \hline \delta_{A | B} & 0 & 1 & 1 & 0 & 0 \\  \delta_{A | C } & 1  & 1 & 0 & 1 & 0 \\ \delta_{A | D } & 1  & 1 & 0 & 0 & 1 \\ \delta_{B | C } & 1  & 0 & 1 & 1 & 0  \\ \delta_{B | D} & 1  & 0 & 1 & 0 & 1 \\ \delta_{C | D } & 0  & 0 & 0 & 1 & 1 \end{matrix} 
$$

Let $\sigma_{AB} = \sigma_{1} r_{AB} + \sigma_{2} \delta_{A | B \cup C \cup D}  +  \sigma_{3} \delta_{B | A \cup C \cup D} +  \sigma_{4} \delta_{C | A \cup B \cup D} +  \sigma_{5} \delta_{D | A \cup B \cup C }$ be a vector in the null space of the matrix $M_{AB}$.  Up to a scalar multiple, 

$$
K_{AB}( \sigma_{AB}) = \left[\begin{matrix} w(AB)_{A|B} \\ w(AB)_{A|C} \\ w(AB)_{A|D}  \\  w(AB)_{B|C} \\  w(AB)_{B|D} \\  w(AB)_{C|D} \end{matrix} \right] =   \left[\begin{matrix}  (a + b) c d (c + d) \\  -b d (b c + a d)  \\ -b c (a c + 
     b d) \\  -a d (a c + b d) \\ -a c (b c + a d) \\  a b (a + b) (c + 
     d) \end{matrix} \right]
$$

Then 

\addvspace{0.1cm}

\begin{equation}
\label{formulaOne}
p_{AB} \cdot p_{AC} = \sum_{ \{U, V \} \in { \{A, B, C, D \} \choose 2}} |U| \cdot |V| \cdot w(AB)_{U|V} \cdot w(AC)_{U|V} 
\end{equation}
and
\begin{equation}
\label{formulaTwo}
||p_{AB}||_{2}  = \sqrt{ \sum_{ \{U, V \} \in { \{A, B, C, D \} \choose 2}} |U| \cdot |V| \cdot [w(AB)_{U|V} ]^{2}}. 
\end{equation}
We use (\ref{formulaOne}) and (\ref{formulaTwo}) to obtain the formulae for the angle measures between the resolution cones.  
\end{proof}

\begin{cor}\label{unrooted:2dim}
The space $\Span \ K_{T}$ is 2-dimensional when $T$ is unrooted.
\end{cor}

\begin{proof}
As in the case of rooted trees, we can use the formula in (\ref{proj:coeffs}) to show that the set $\{p_{AB}, p_{AC} \}$ is linearly independent, but the set $\{p_{AB}, p_{AC}, p_{BC} \}$ is linearly dependent. 
\end{proof}

Corollary \ref{unrooted:2dim} also follows from the fact that the set of all arbitrary tree metrics is a tropical variety in $\rr^{n(n-1)/2}$, as shown in \cite{TropGrass}.


\section{Distance-Based Methods Near a Tritomy}\label{sec:tritomymethods}

In this section, we analyze the performance of distance-based methods
around a tritomy using the results on the geometry of tree space
from Section \ref{sec:tritomygeometry}.  The basic observation
is this:  any phylogenetic algorithm decomposes the set of
all dissimilarity maps into regions based upon which combinatorial
type of tree gets reconstructed by the algorithm.  While in many 
cases we do not have a complete understanding of the geometry
of these decompositions across all of $\rr^{n(n-1)/2}_{\geq 0}$,
we can describe the geometry in a small neighborhood of a tree metric
with a single tritomy.  It is this geometry which we
explore in the present section.


\subsection{Least Squares Phylogeny}

Let $C_{1}, \ldots, C_{r}$ be subsets of $\rr^{n(n-1)/ 2}$.    
The \textit{Voronoi cell} $V_{k}$ associated with the subset $C_{k}$ is the set of all points
$$
V_{k}  = \{ \mathbf{x} \in \rr^{n(n-1) / 2} \ | \ d(\mathbf{x}, C_{k} ) \leq d(\mathbf{x}, C_{j} ) \ \text{for \ all } j \neq k \}
$$
where $d(\mathbf{x}, C_{k})= \inf\{ || \mathbf{x}-  \mathbf{a} ||_{2} \ | \ \mathbf{a} \in C_{k} \}$.
The \textit{Voronoi decomposition} is the subdivision of $\rr^{n(n-1)/ 2}$ into Voronoi cells of the set $\{ C_{k} \}$.  When $\{ C_{k} \}$ is the collection of
cones associated to all possible combinatorial trees with leaf set $[n]$, the Voronoi cells 
comprise the subdivision of space induced by the least squares phylogeny problem. 

While the Voronoi decomposition of a finite set of points is
well-known to be a polyhedral subdivision of space, the Voronoi decomposition
induced by a collection of higher dimensional polyhedra can be a 
complicated semi-algebraic decomposition.  Hence, the Voronoi decomposition
induced by the tree cones is probably not polyhedral.
We saw in Section 3 that in a neighborhood of a tree metric $T$ with a single tritomy,  tree space has the form $\rr^{k}  \times K_{T}$, where $k = \dim \  \Span \ C_{T}$ and 
$K_{T}$ is a one-dimensional fan with three rays that sits naturally inside a two dimensional linear space $\Span \  K_{T}$.  In this setting it is easy to describe the Voronoi
decomposition.


\begin{prop}\label{prop:openbook}
Let $T$ be a tree metric in $\rr^{n(n-1)/ 2} _{\geq 0}$ with local tree space $\rr^{k} \times K_{T}$. 
The boundary between the Voronoi cells for the resolution cones $C_{T_{AB}}$ and $C_{T_{AC}}$ are completely determined by the angle bisector in $\Span \ K_{T}$  between $p_{AB}$ and $p_{AC}$.
\end{prop}

\begin{proof}
The Euclidean distance between the cones $C_{T_{AB}}$ and $C_{T_{AC}}$ is the sum of the distance between them in the two orthogonal spaces $\Span \ C_{T}$ and $\Span \ K_{T}$.  Of these two distances, only the distance in the two-dimensional space $\Span \ K_{T}$ is nonzero; this distance is determined by the angle between the maximal cones in the 1-dimensional polyhedral fan $K_{T}$.  The set of all points in the plane $\Span \ K_{T}$ equidistant between two vectors emanating from the origin is the bisector of the angle between the two vectors.  
\end{proof}

Proposition \ref{prop:openbook} allows us to easily compute the
relative size of the Voronoi regions around a polytomy for
either equidistant or ordinary tree metrics.  The next theorem gives a formula for the boundary between the Voronoi regions. 

\begin{thm}\label{equidistant_bisector}
Let $T$ be a ranked, rooted tree with a single tritomy.  The boundary of the Voronoi cell in $\Span \ K_{T}$ between the resolution cones $C_{T_{AB}}$ and $C_{T_{AC}}$ is spanned by the vector  
$$
 \frac{p_{AB}}{\sqrt{a + b}} +  \frac{ p_{AC}}{\sqrt{a + c}}.
 $$
\end{thm}

\begin{proof}
By Proposition \ref{prop:openbook} the boundary of the cell we wish to compute is given by by the angle bisector in $\Span \ K_{T}$ between $p_{AB}$ and $p_{AC}$, which is spanned by the normalized average of the two vectors.  By  Lemma \ref{outlier_projection}  we have
$$
\| p_{AB} \|_{2}  = \sqrt{ \frac{ abc(a+b)}{ab + ac+bc}}  \mbox{ and  }  
\| p_{AC} \|_{2}  = \sqrt{ \frac{ abc(a+c)}{ab + ac+bc}}
$$
Therefore
$$
\frac{p_{AB}}{\|p_{AB}\|} + \frac{p_{AC}}{\|p_{AC}\|} =  \sqrt{ \frac{ab + ac+bc}{abc}} \left( \frac{p_{AB}}{\sqrt{a + b}} +  \frac{ p_{AC}}{\sqrt{a + c}} \right)
$$
So $ \frac{p_{AB}}{\sqrt{a + b}} +  \frac{ p_{AC}}{\sqrt{a + c}}$  spans the boundary of the Voronoi cells for the two cones in $\Span \ K_{T}$.  
\end{proof}

\begin{thm}\label{unrooted_bisector} 
Let $T$ be an unrooted tree with a single tritomy.  The boundary of the Voronoi cell in $\Span \ K_{T}$ between the resolution cones $C_{T_{AB}}$ and $C_{T_{AC}}$ is spanned by the vector 
$$
\frac{p_{AB}}{\sqrt{(a + b)(c + d)} } + \frac{p_{AC}}{\sqrt{(a + c)(b + d)}}.
$$
\end{thm}

\begin{proof}
We use the fact that $p_{AB}$ and $p_{AC}$ are constant on the vectors $\delta_{U|V}$ for $\{U, V\} \in { \{A, B, C, D\} \choose 2}$ and the formulas for the coeffcients $w(A, B) _{U|V}$ in the proof of Theorem \ref{unrooted_formulas}  to calculate the 2-norms of $p_{AB}$ and $p_{AC}$. Up to an identical polynomial $f$ in the variables $a,b,c$ and $d$, we have
\begin{equation}\label{unrooted:factor:one}
\| p_{AB} \|_{2} = f \cdot \sqrt{(a + b)(c + d)}
\end{equation}
and 
\begin{equation}\label{unrooted:factor:two}
 \| p_{AC} \|_{2} = f \cdot \sqrt{(a + c)(b + d)} 
\end{equation}
As in Theorem \ref{equidistant_bisector} we know that the angle bisector between $p_{AB}$ and $p_{AC}$ gives the boundary of the Voronoi cell in $\Span \ K_{T}$.  By  (\ref{unrooted:factor:one}) and (\ref{unrooted:factor:two})  the angle bisector is a multiple of 
$$
\frac{p_{AB}}{\sqrt{(a + b)(c + d)} } + \frac{p_{AC}}{\sqrt{(a + c)(b + d)}}.
$$

\end{proof}


\subsection{UPGMA Regions Near a Polytomy}

The UPGMA algorithm (Unweighted Pair Group Method with Arithmetic Mean) \cite{Sokal1963} is an agglomerative tree reconstruction method that takes as an input $\alpha$, a dissimilarity map of ${n \choose 2} = n(n-1)/2$ pairwise distances between a set $X$  of $n$ taxa, and returns a rooted equidistant tree metric $d$ on $X$.  The metric $d$ is an approximation to the least squares phylogeny for $\alpha$.  In this section we show that in some circumstances, UPGMA fails to correctly identify the least squares phylogeny.  The occurrence and severity of this failure depends entirely on the relative sizes of the daughter clades $A$, $B$, and $C$ of the tritomy.  The algorithm works as follows:

\begin{algorithm}
\caption{UPGMA}
\label{UPGMA}


\begin{itemize}

\item Input: a dissimilarity map $\alpha \in \rr^{ {n(n-1) / 2}}_{\geq 0}$ on $[n]$.
\item Output: a maximal  chain $C$ in the partition lattice $\Pi_{n}$ and an 
equidistant tree metric $d$.
\item Initialize $\pi_{n} = 1|2| \cdots | n$, and set $\alpha^{n} = \alpha$.
\item For $i = n-1, \ldots, 1$ do
\begin{itemize}
\item  From partition $\pi_{i+1} = \lambda^{i+1}_{1} | \cdots | \lambda^{i+1}_{i+1}$
and distance vector $\alpha^{i+1} \in \rr^{(i+1)i/2}_{\geq 0}$
choose $j, k$ be so that $\alpha^{i+1}(\lambda^{i+1}_{j}, \lambda^{i+1}_{k})$ is minimized.
\item  Set $\pi_{i}$ to be the partition obtained from $\pi_{i+1}$ by
 merging $\lambda^{i+1}_{j}$ and $ \lambda^{i+1}_{k}$ and leaving all other
 parts the same. Let $\lambda^{i}_{i} = \lambda^{i+1}_{j} \cup \lambda^{i+1}_{k}$.
\item  Create new distance $\alpha^{i} \in \rr^{i(i-1)/2}_{\geq 0}$ by
$\alpha^{i}(\lambda, \lambda') = \alpha^{i+1}(\lambda, \lambda')$ if $\lambda, \lambda'$ are
both parts of $\pi_{i+1}$ and
$$
\alpha^{i}(\lambda, \lambda^{i}_{i})
=  \frac{  |\lambda^{i+1}_{j}|}{ |\lambda^{i}_{i}|} 
\alpha^{i+1}( \lambda, \lambda^{i+1}_{j} )  + 
\frac{|\lambda^{i+1}_{k}|}{ |\lambda^{i}_{i}|} 
\alpha^{i+1}( \lambda, \lambda^{i+1}_{k} )
$$
otherwise.
\item  For each $x \in \lambda^{i+1}_{j}$ and $y \in\lambda^{i+1}_{k}$, 
set $d(x,y) = \alpha^{i+1}(\lambda^{i+1}_{j}, \lambda^{i+1}_{k})$.
\end{itemize}
\item Return:  Chain $C = \pi_{n} \lessdot \cdots \lessdot \pi_{1}$ and
equidistant metric $d$.
\end{itemize}

\end{algorithm}

Note that if the blocks $A, B \subset [n]$ are joined in step $i$ of Algorithm \ref{UPGMA} the distance recalculation implies
$$
\alpha^{i}(A, B)  = \frac{1}{|A| |B|} \sum_{x \in A, y \in B} \alpha^{n} (x,y) = \frac{1}{ab} \sum_{x \in A, y \in B} \alpha (x,y),
$$
a formula which is useful in the next Proposition:

\begin{prop}\label{perp_calc}
Let $T$ be a ranked, rooted tree with a single tritomy.  The boundaries between UPGMA regions in $\rr^{n(n-1)/2}_{\geq 0}$ for the resolution cones  $C_{T_{AB}}$, $C_{T_{AC}}$ and $C_{T_{BC}}$ are orthogonal to the plane $\Span \ K_{T}$.  
\end{prop}

\begin{proof}
The boundary between the UPGMA regions for the cones $C_{T_{AC}}$ and $C_{T_{BC}}$ is given by the condition 
$$
\alpha^{k}(A,C) = \alpha^{k} (B,C)
$$
which translates into the following linear condition on the original
dissimilarity map
$$
\frac{1}{ac}  \sum_{i \in A, j \in C} \alpha(i,j) =
\frac{1}{bc}  \sum_{i \in B, j \in C} \alpha(i,j). 
$$
This hyperplane  has normal vector 
$$
\frac{1}{ac}  \delta_{A|C} -\frac{1}{bc}  \delta_{B|C}.
$$
Now 
$$
-\frac{1}{ac} p_{AC} = \frac{1}{ab + ac + bc} \left ( -\delta_{A|B} - \delta_{B|C} + \frac{ab + bc}{ac} \delta_{A|C} \right)
$$
and 
$$
\frac{1}{bc} p_{BC} = \frac{1}{ab + ac + bc} \left ( \delta_{A|B} + \delta_{A|C} - \frac{ab + ac}{bc} \delta_{B|C} \right)
$$
So 
$$
-\frac{1}{ac} p_{AC} + \frac{1}{bc} p_{BC} =  \frac{1}{ab + ac + bc} \left( \left ( \frac{ab + bc}{ac} + 1 \right) \delta_{A|C} - \left(\frac{ab + ac}{bc} + 1 \right) \delta_{B|C} \right)
$$
$$
= \frac{1}{ac}  \delta_{A|C} -\frac{1}{bc}  \delta_{B|C}.
$$
Thus, the normal vector for the boundary between UPGMA regions for the cones $C_{T_{AC}}$ and $C_{T_{BC}}$ is in $\Span \ K_{T}$, and the boundary is orthogonal to $\Span \ K_{T}$.  The calculation is the same for the other two pairs of cones.  
\end{proof}


\begin{thm}\label{UPGMA_boundary}
The boundary between the UPGMA cells for the resolution tree topologies $T_{AC}$ and $T_{BC}$ in $\Span \ K_{T}$ is $-p_{AB}$.  
\end{thm}

\begin{proof}

Since $\Span \ K_{T}$ is two-dimensional and the boundaries between the UPGMA regions for the resolutions $T_{AB}$, $T_{AC}$ and $T_{BC}$ are orthogonal to $K_{T}$ by Proposition \ref{perp_calc}, it suffices to find a vector $\omega \in \Span  \ K_{T}$ that satisfies 
$$
\frac{1}{a c } \sum_{i \in A, j \in C} \omega_{i,j} = \frac{1}{bc } \sum_{k \in B, \ell \in C} \omega_{k,\ell} \leq \frac{1}{ab } \sum_{m \in A, n \in B} \omega_{m,n}.
$$
Any such vector will span the boundary of the UPGMA cells.   We have
$$
\frac{1}{a c } \sum_{i \in A, j \in C} -(p_{AB})_{i,j} =  \frac{1}{bc } \sum_{k \in B, \ell \in C} -(p_{AB})_{k,\ell} = -\frac{ab}{ab + ac + bc}
$$
while
$$
\frac{1}{ab } \sum_{m \in A, n \in B} -(p_{AB})_{m,n} = \frac{ac + bc}{ab + ac + bc}
$$
So $-p_{AB}$ satisfies the required condition.  
\end{proof}

 \subsection{UPGMA and LSP Cells}
 
  In this section  we discuss how results from Sections \ref{sec:tritomygeometry} and \ref{sec:tritomymethods} show that UPGMA poorly matches LSP in some circumstances. The geometry of the fan $K_{T}$ and the UPGMA cells in $\Span \ K_{T}$ for equidistant trees depends entirely on the size of the daughter clades $A, B$, and $C$ of the tritomy.  Consequentially, the quality of the performance of UPGMA near a tree metric with a tritomy depends on how similar in size the daughter clades are.  When $a = b = c$, the UPGMA and LSP regions near a tritomy are the same, but as either one or two of the daughter clades becomes much larger, UPGMA does a poorer job of identifying the LSP.     We also use results from Section \ref{sec:tritomymethods}  to show that NJ poorly matches LSP in specific examples for small numbers of taxa.

We can use our theorems about the geometry of $\Span \ K_{T}$ to  investigate the relative size of the UPGMA and Voronoi cells as $a, b$, and $c$ vary.  By Theorem \ref{equidistant:angle}  the angle between $C_{T_{AC}}$ and $C_{T_{BC}}$ is 
$$
 \arccos \left(   \frac{-c}{\sqrt{(a + c)(b + c)} }\right) 
$$
and this is also the angle measure of the UPGMA region associated with the cone
$C_{T_{AB}}$.  By the angle bisector argument, we see that the angle
measure of the LSP region associated to the tree $T_{AB}$ near the
tritomy will be:
$$
 \frac{1}{2} \arccos \left(   \frac{-a}{\sqrt{(a + b)(a + c)} }\right) + 
 \frac{1}{2} \arccos \left(   \frac{-b}{\sqrt{(a + b)(b + c)} }\right). 
$$
When $c > > a \approx b$, the angle for the UPGMA region approaches $\pi$ whereas
the angle for the LSP region approaches $\pi/{2}$.
Conversely, when $a \approx b  >> c$ the angle for the UPGMA region approaches
$\pi/2$ whereas the angle for the LSP region approaches $3 \pi/4$.
Tables \ref{table:CBiggerAB} and \ref{table:ABBiggerC} compare the sizes of the various regions for
differing values of $a,b,$ and $c$.  We display the sizes as the percentage of the total amount of the local volume around the polytomy that corresponds to the UPGMA or LSP region for the cone $C_{T_{AB}}$.

\begin{table}
\renewcommand{\arraystretch}{1.3}
\caption{Region sizes for $C_{T_{AB}}$ when $c > > a = b$}
\label{table:CBiggerAB}
\centering
\begin{tabular}{| c | c | c | c | c |}
\hline
a & b & c & UPGMA  & LSP  \\ \hline
1 & 1 & 1 & 33.3333 & 33.3333  \\
1 & 1 & 2 & 36.6139 & 31.693  \\
1 & 1 & 4 & 39.7583 & 30.1209 \\
1 & 1 & 8 & 42.4261 & 28.787 \\
1 & 1 & 16 & 44.5139 &  27.7431 \\
1 & 1 & $2^{10}$ & 49.297 & 25.3515  \\
1 & 1 & $2^{20}$ & 49.978 & 25.011  \\ \hline
\end{tabular}
\end{table}

\begin{table}
\renewcommand{\arraystretch}{1.3}
\caption{Region sizes for $C_{T_{AB}}$ when $a = b  > > c $}
\label{table:ABBiggerC}
\centering
\begin{tabular}{| c | c | c | c | c |}
\hline
a & b & c & UPGMA  & LSP    \\ \hline
1 & 1 & 1 &  33.3333 & 33.3333  \\
2 & 2 & 1 & 30.4086 & 34.7957 \\
4 & 4 & 1 & 28.2046 & 35.8977 \\
8 & 8 & 1 & 26.7721 & 36.6139 \\
16 & 16 & 1 & 25.9367 & 37.0317 \\
$2^{10}$ & $2^{10}$ & 1 & 25.0155 & 37.4923  \\
$2^{20}$ & $2^{20}$ & 1 & 25.0001 & 37.4999  \\ \hline
\end{tabular}
\end{table}

While the convergence to the limiting values is slow, already for small values of $a,b$, and
$c$ there is significant discrepancy between UPGMA and LSP.
Figures \ref{CBIGAB} and \ref{ABBIGC} illustrate  the geometry of this phenomenon for the two extreme cases $c > > a \approx b$, and $a \approx b > > c$.  In both figures the fan $K_{T}$ is black, the vector $p_{AB}$ is labeled with the pair $AB$, LSP boundaries are blue, and UPGMA boundaries are red.  Note that when $c > > a \approx b$, UPGMA overestimates the size of the LSP region for $C_{T_{AB}}$.  When $a \approx b > > c$, UPGMA underestimates the size of the LSP region for $C_{T_{AB}}$.   
 
\begin{figure}
\centering
\includegraphics[width=8cm]{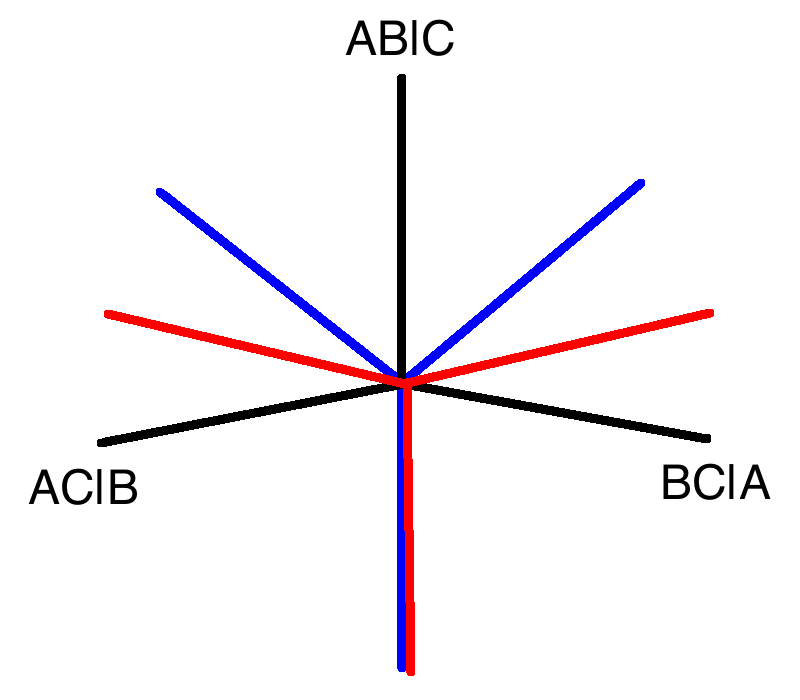} 
 \caption{The case $c > > a \approx b$. The fan $K_{T}$ is black, LSP boundaries are blue,  and UPGMA boundaries are red.}
 \label{CBIGAB}
\end{figure}

\begin{figure}
\centering
\includegraphics[width=8cm]{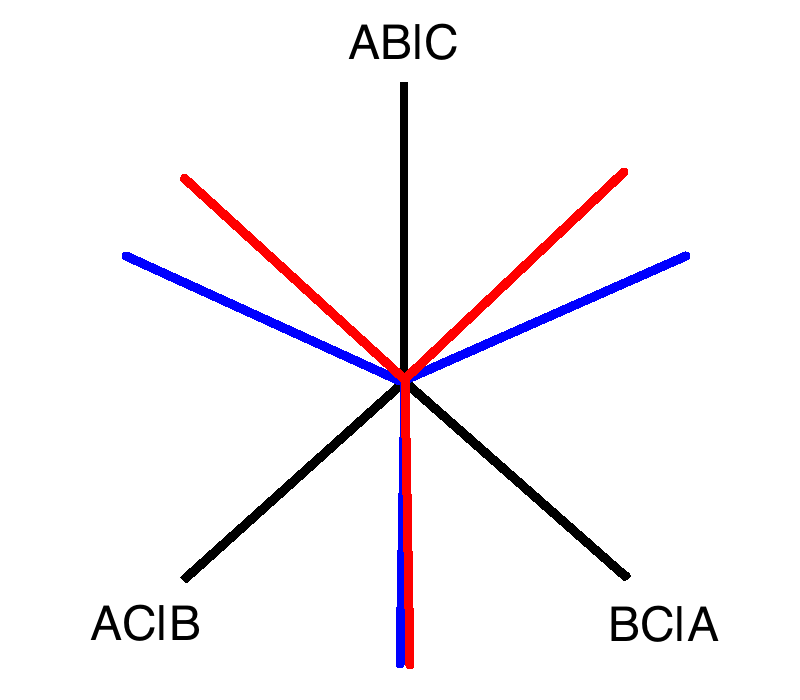} 
 \caption{The case $a \approx b  > > c$. The fan $K_{T}$ is black, LSP boundaries are blue,  and UPGMA boundaries are red. }
 \label{ABBIGC}
\end{figure}


\subsection*{LSP Cells and Local NJ Behavior}

The Neighbor-Joining (NJ) algorithm, due to Saitou and Nei \cite{SaitouNei} is a distance-based reconstruction method that returns an unrooted tree $T$ and a tree metric realized by $T$.  Both the selection criterion (known as the "$Q$-criterion") and distance recalculation are linear combinations of the original input coordinates.  Therefore, as in the case of UPGMA, NJ divides the input space $\rr^{n(n-1)/2}_{\geq 0}$ into a family of polyhedral cones studied in \cite{NJSphere} and \cite{Haws2011}.  

A complete combinatorial description of the NJ cones remains unknown, and we do not have a closed description of the local geometry of the NJ regions around a tritomy.  However, by running NJ on points sampled uniformly from the surface of  a small sphere around a tritomy, we can obtain an empirical estimate of the local relative size of NJ regions for small numbers of taxa.  

For unrooted tree metrics, the case of interest is when $a$ and $b$ are larger than $c$ and $d$: if $a = b = c$ and $d$ is larger or smaller, the size of the LSP cells the for three resolution cones will be symmetric.  NJ poorly identifies LSP when $a$ and $b$ are larger than $c$ and $d$ even for small numbers of taxa.  Unlike in the case of UPGMA, the relative size of the regions appears to be dictated not only by $a, b, c$, and $d$, but also by the topology of the subtrees with leaf sets $A, B, C$, and $D$.   

\begin{algorithm}
\caption{Neighbor-Joining}
\label{NJ}

\begin{itemize}

\item Input: a dissimilarity map $\alpha \in \rr^{ {n(n-1)/2}}_{\geq 0}$ on $[n]$.
\item Output: an unrooted binary tree $T$ and a tree metric $d_{T, w} = d$ realized by $T$.  
\item Initialize $[ n] = \{1,2,... , n\}$, and set $d_{0} = \alpha$.
\item For $r = 1 , \ldots, n-3$ do
\begin{itemize}
\item Identify subsets $A_{i}, A_{j}$ of $[n]$   minimizing $$Q_{r}(A_{i}, A_{j}) = (n-r-1)d_{r-1}(A_{i}, A_{j}) - \sum_{k= 1}^{n-r +1} d_{r}(A_{i}, A_{k}) - \sum_{k = 1}^{n-r + 1} d_{r}(A_{j}, A_{k})$$, \item Update $$d_{r}(A_{ij}, A_{k}) = \frac{1}{2} \left( d_{r-1}(A_{i}, A_{k}) +   d_{r-1}(A_{j}, A_{k}) - d_{r}(A_{i}, A_{j}) \right) $$

\end{itemize}

\item Return:  unrooted binary combinatorial tree $T$, $w : E(T)  \rightarrow \rr$ and
tree metric $d_{n-3} = d_{T, w}$.
\end{itemize}

\end{algorithm}

Applying Theorem \ref{unrooted_formulas}, we see that when $a \approx b >  > c \approx d$, the angle between the cones $C_{T_{AC}}$ and $C_{T_{AD}}$ approaches $\pi$, while the LSP angle for $C_{T_{AB}}$, bounded by angle bisectors between the two pairs $\{p_{AB}, p_{AC}\}$ and $\{p_{AB}, p_{AD}\}$, approaches $\pi/2$.  Figure \ref{ABBiggerCD} shows this case.  The fan $K_{T}$ is black, and the LSP boundaries are blue.

\begin{figure}
\centering
\includegraphics[width=8cm]{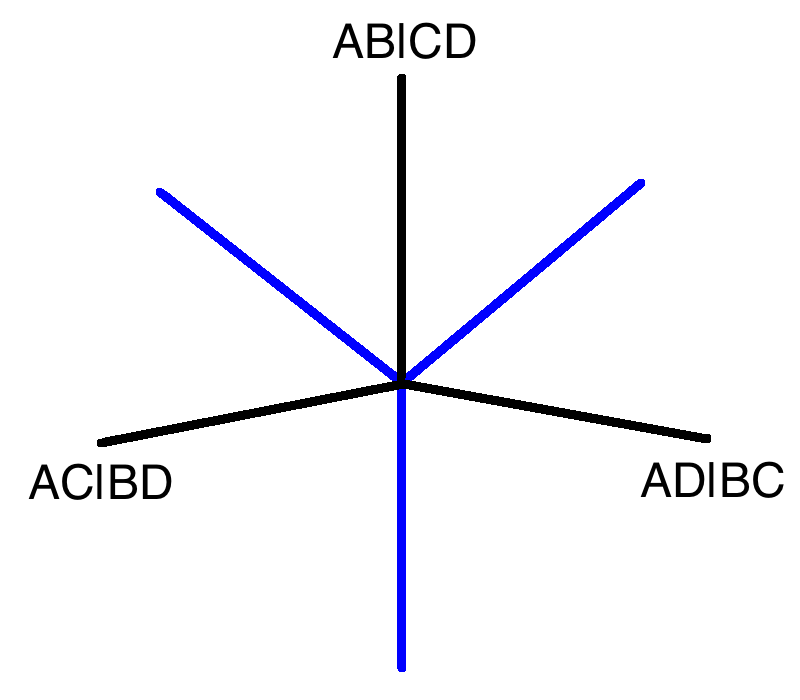}
\caption{The case $a \approx b > > c \approx d$. The fan $K_{T}$ is black, and LSP boundaries are blue. }
\label{ABBiggerCD}
\end{figure}

For small values of $a,b,c$ and $d$ we present computational evidence that NJ fails to identify LSP correctly.  Consider the tree metrics $d_{1}$  and $d_{2}$ with topologies shown in Figures \ref{ETomy} and \ref{FTomy} and given edge weights of randomly assigned numbers between 5000 and 10000.  Here $A = \{1,2,3,4,5,6 \}, B = \{7,8,9,10,11,12 \}, C = \{13\}$, and  $D = \{14\}$.  Using Theorem \ref{unrooted_formulas} we can calculate the angles between the pairs of cones in $\{C_{T_{AB}}, C_{T_{AC}}, C_{T_{AD}} \}$ and find the relative sizes of the LSP regions near the tritomies $d_{1}$ and $d_{2}$.  Recall that one consequence of Theorem \ref{unrooted_formulas} is that the relative size of the LSP regions near $d_{1}$ and $d_{2}$ will be the same because these proportions only depend on $a, b, c$, and $d$.

\begin{figure}
\centering
\includegraphics[width=8cm]{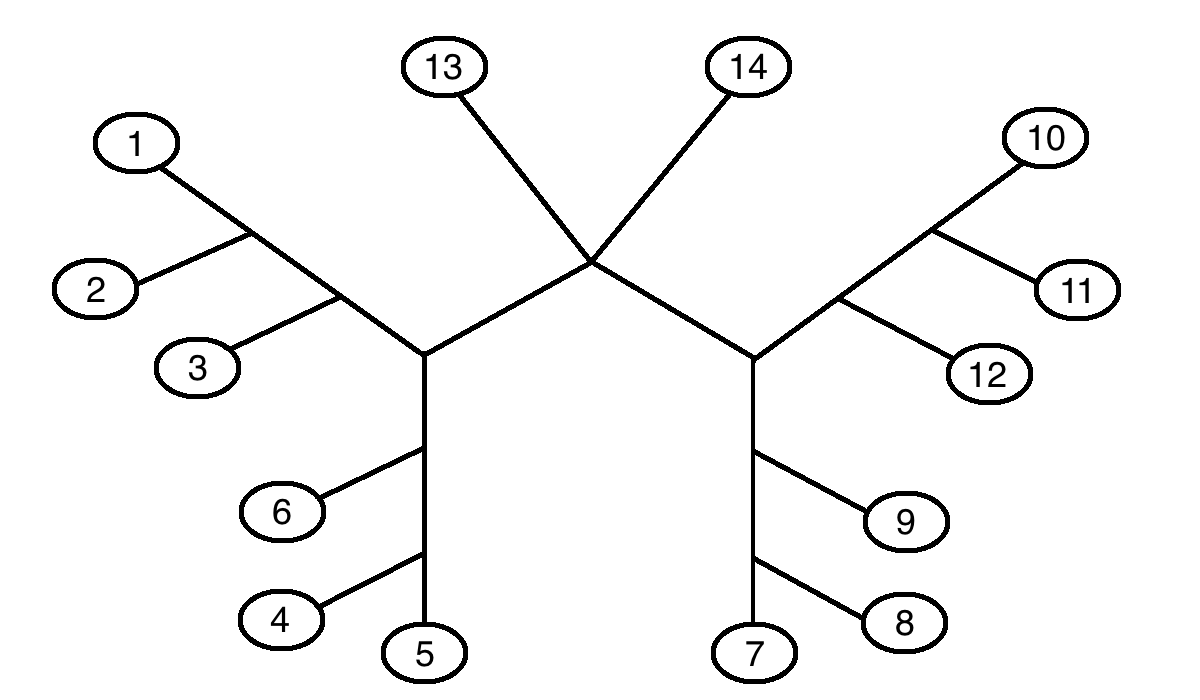}
\caption{A tritomy $d_{1}$ on the leaf set $[14]$}
\label{ETomy}
\end{figure}

\begin{figure}
\centering
\includegraphics[width=8cm]{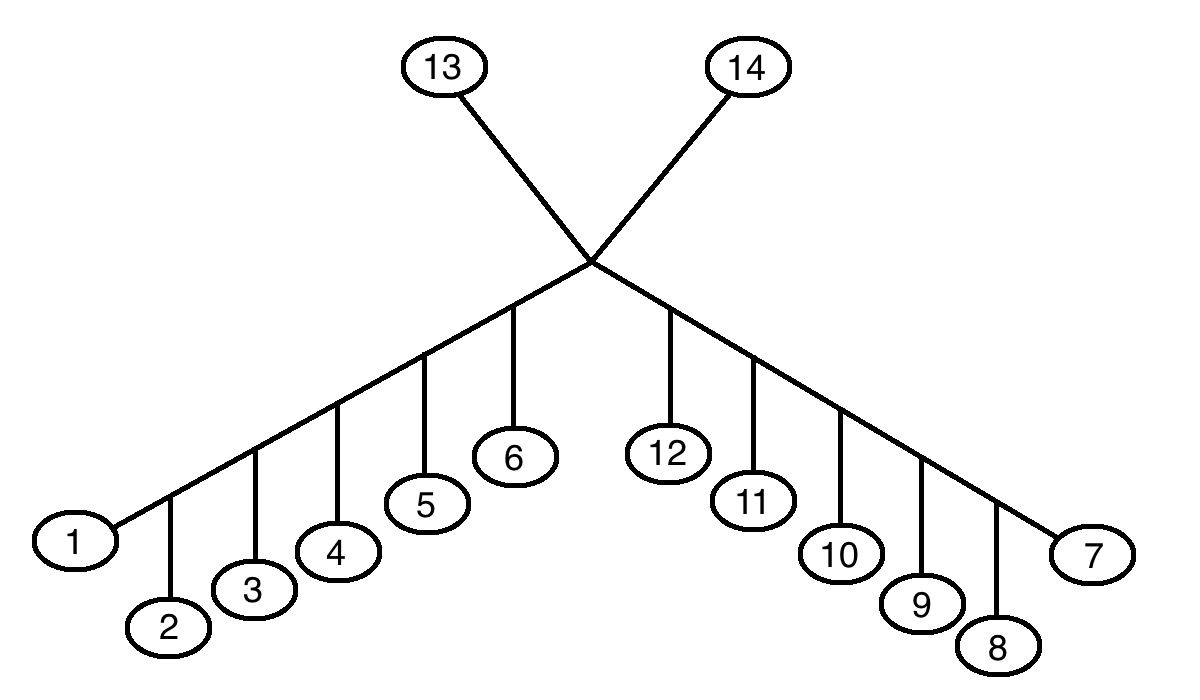}
\caption{A tritomy $d_{2}$ on the leaf set $[14]$}
\label{FTomy}
\end{figure}

Running NJ on 1,000,000 points sampled uniformly from spheres of radius 0.05 centered at $d_{1}$ and $d_{2}$ gives an empirical measure of the size of NJ regions for the three resolutions $T_{AB}$, $T_{AC}$ and $T_{AD}$ near the two points. We compare this empirical distribution with the size of the LSP regions computed via Theorem 3.7 in Table \ref{NJTable}.   Sizes of the regions are given as percentages of the total local volume near the tritomy.  

\begin{table}
\renewcommand{\arraystretch}{1.3}
\caption{ NJ and LSP near $d_{1}$ and $d_{2}$}
\label{NJTable}
\centering
\begin{tabular}{ | c | c | c | c |  }
\hline
 Resolution of Splits &  LSP & NJ : $d_{1}$ & NJ: $d_{2}$   \\ \hline
 \quad $ AB | CD $ \quad  & \quad 30.6897 \quad & \quad  38.1501 \quad  &  \quad 35.7037 \\
\quad $AC | BD$ \quad & \quad  34.6552 \quad & \quad  30.9344 \quad & \quad 32.1305 \\
 \quad $AD | BC $ \quad & \quad 34.6552 \quad & \quad  30.9155 \quad & \quad 32.1658 \\
 \hline
\end{tabular}
\end{table}

\addvspace{1pc}

Table \ref{NJTable} shows that NJ overestimates the size of the LSP regions near $d_{1}$ and $d_{2}$ closest to the cone $C_{T_{AB}}$ and underestimates the regions near $C_{T_{AC}}$ and $C_{T_{AD}}$.  Furthermore, the topological structure of the subclades $A$ and $B$ influence the local size of the NJ regions.  This shows that a direct analog of Theorem \ref{UPGMA_boundary} will not exist for NJ.   However, there may exist an analogous theorem for NJ when the topology of the subtrees around the polytomy is taken into account.

\section{Conclusion} \label{sec:conclusion}

Distance-based heuristics like UPGMA and NJ can be seen as approximating solutions to the intuitively appealing but NP-hard least-squares phylogeny problem.  We compared heuristics to LSP when the true tree metric contains a tritomy.  For UPGMA, our theoretical analysis shows that the success rate of the heuristic greatly depends on how balanced the sizes of the underlying daughter clades are.

Due to the precise form of input data used by NJ and the outputs of the algorithm, NJ is an approximation to LSP.  However, Gascuel and Steel showed that  NJ performs a heuristic search, guided by the $Q$-criterion at each agglomeration step, that minimizes a tree-length estimate due to Pauplin known as the "Balanced Minimum Evolution" (BME) criterion (\cite{Pauplin}, \cite{GascuelSteel}).  This insight was incorporated into the selection criterion and distance recalculation aspects of the  algorithms  BIONJ \cite{BIONJ}, Weighbor  \cite{Weighbor}, and FastME \cite{DesperGascuel}.  These algorithms take distance matrices as input and have superior performance to NJ in terms of topological accuracy and better immunity to pathologies such as the long-branch attraction.  However, the subdivision of the input spaces induced by each of these improved algorithms is not polyhedral and, like the Voronoi cells around higher-degree polytomies, have a complicated semi-algebraic description.  

Any improvements to distance-based methods implied by the results in this paper would require a fundamentally different approach, such as changing the $Q$-criterion at each step to reflect the size of the taxon groups to be joined.



\section*{Acknowledgments}

Ruth Davidson was partially supported by the US National Science Foundation (DMS 0954865).
Seth Sullivant was partially supported by the David and Lucille Packard 
Foundation and the US National Science Foundation (DMS 0954865).


\end{document}